\documentclass[12pt, a4paper, twoside]{article}
\usepackage[utf8]{inputenc}
\usepackage{fullpage}
\usepackage{vivek_latex_macros}
\usepackage{amsthm,amssymb,amsmath}
\usepackage[ruled]{algorithm2e}

\newtheorem{theorem}{Theorem}[section]
\newtheorem{corollary}[theorem]{Corollary}

\newtheorem*{problem}{Problem L}
\usepackage{cite}
\renewcommand{\baselinestretch}{1.07}

\title{A Las Vegas algorithm to solve the elliptic curve discrete logarithm problem}
\author{Ayan Mahalanobis\thanks{ayan.mahalanobis@gmail.com} \and Vivek Mallick\thanks{vmallick@iiserpune.ac.in}}
\newcommand{\Address}{{% additional braces for segregating \footnotesize
\bigskip
\footnotesize
\hrulefill\;
\textsc{IISER Pune, Pashan, Pune, INDIA}\par\nopagebreak
}}

\date{\today}
\begin{document}
\maketitle
\begin{abstract}
In this paper, we describe a new Las Vegas algorithm to solve the elliptic curve discrete logarithm problem. The algorithm depends on a property of the the group of rational points of an elliptic curve and is thus not a generic algorithm. The algorithm that we describe has some similarities with the most powerful index-calculus algorithm for the discrete logarithm problem over a finite field.
\end{abstract}
\section{Introduction}
Public-key cryptography is a backbone of this modern society. Many of the public-key cryptosystems depend on the \emph{discrete logarithm problem} as their cryptographic primitive. Of all the groups used in a discrete logarithm based protocol, the group of \emph{rational points of an elliptic curve} is the most popular. In this paper, we describe a \textbf{Las Vegas algorithm} to solve the elliptic curve discrete logarithm problem.

There are two kinds of attack on the discrete logarithm problem. One is generic. This kind of attack works in any group. Examples of such attacks are the baby-step giant-step attack~\cite[Proposition 2.22]{silverman} and Pollard's rho~\cite[Section 4.5]{silverman}. The other kind of attack depends on the group used. Example of such attack is the index-calculus attack~\cite[Section 3.8]{silverman} on the multiplicative group of a finite field. An attack similar to index calculus for elliptic curves, known as xedni calculus, was developed by Silverman~\cite{xedni1,xedni2}. However, it was found to be no better than exhaustive search. Another simailar work in the direction of ours is Semaev~\cite{semaev}. 

In this paper, we describe an attack which is particular to the elliptic curves. The attack is a Las Vegas algorithm. The attack uses a theorem for elliptic curve. The idea behind the attack is completely new and is of a completely different genre from the existing ones~\cite{gaudry,lange,gailbraith1,gailbraith2}. In comparison to xedni calculus, our algorithm is fairly straightforward to understand, implement and is better than the exhaustive search. 

The main algorithm is divided into two algorithms. The first one reduces the elliptic curve discrete logarithm problem to a problem in linear algebra. We call the linear algebra problem, Problem L. This reduction is a Las Vegas algorithm with \textbf{probability of success} $0.6$ and is \textbf{polynomial} in both time and space complexity. The second half of the algorithm is solving Problem L. This is the current bottle-neck of the whole algorithm and better algorithms to solve Problem L will produce better algorithms to solve elliptic curve discrete logarithm problem. The success of the main algorithm is $0.6\times\left(\log{p}\right)^2/p$ where every pass is polynomial time in time and space complexity.
	\subsection{The central idea behind our attack}
Let $G$ be a cyclic group of prime order $p$. Let $P$ be a non-identity element and $Q(=mP)$ belong to $G$. The \emph{discrete logarithm problem} is to compute the $m$. One way to find $m$ is to find integers $n_i$, for $i=1,2,\ldots,k$ for some positive integer $k$ and  $1\leq n_i <p$ such that $\sum_{i=1}^k n_i=m\bmod{p}$. The last equality is hard to compute because we do not know $m$. However we can decide whether 
\begin{equation}\label{eqn1}
\sum_{i=1}^k n_iP=Q
\end{equation} 
and based on that we can decide if $\sum_{i=1}^k n_i=m\bmod{p}$. Once the equality  holds, we have found $m$ and the discrete logarithm problem is solved.

The number of possible choices of $n_i$ for a given $k$ that can solve the discrete logarithm problem is the number of partitions of $m$ into $k$ parts modulo a prime $p$. The applicability of the above method depends on, how fast can one decide on the equality in the above equation and on the probability, how likely is it that a given set of positive integers $n_i$ sums to $m\bmod{p}$? 

An obvious question is raised, can one choose a set of $n_i$ in such a way that the probability of an equality is higher than the random selection? In the next section, we find a way to check for equality in the case of elliptic curves, however our choice of $n_i$ is uniformly random. Then the algorithm is somewhat straightforward, fix a $k$, choose $n_i$ uniformly random and then check for equality. Once there is a set of $n_i$ for which the equality is found, we have solved the discrete logarithm problem.
\section{The elliptic curve discrete logarithm problem}
The elliptic curve discrete logarithm problem (ECDLP) is the heart and soul of modern public-key cryptography. This paper is about a new probabilistic algorithm to solve this problem. Our algorithm is a fairly straightforward application of the Riemann-Roch theorem. We denote by $\mathcal{E}(\mathbb{F}_q)$ the group of rational points of the elliptic curve $\mathcal{E}$ over $\mathbb{F}_q$. It is well known that there is an isomorphism $\mathcal{E}(\mathbb{F}_q)\rightarrow \text{Pic}^0(\mathcal{E})$ given by $P\mapsto [P]-[\mathcal{O}]$~\cite[Proposition 4.10]{milne}.
\begin{theorem}\label{thm1}
Let $\mathcal{E}$ be an elliptic curve over $\mathbb{F}_q$ and $P_1,P_2,\ldots,P_k$ be points on that curve, where $k=3n^\prime$ for some positive integer $n^\prime$. Then $\sum_{i=1}^kP_i=\mathcal{O}$ if and only if there is a curve $\mathcal{C}$ of degree $n^\prime$ that passes through these points. Multiplicities are intersection multiplicities.
\end{theorem}
\begin{proof}
Assume that $\sum_{i=1}^kP_i=\mathcal{O}$ in $\mathbb{F}_q$ and then it is such in the algebraic closure $\mathbb{\bar{F}}_q$. From the above isomorphism, $\sum_{i=1}^kP_i\mapsto\sum_{i=1}^k[P_i]-k[\mathcal{O}]$. Then $\sum_{i=1}^k[P_i]-k[\mathcal{O}]$ is zero in the Picard group $\text{Pic}^0_{\mathbb{\bar{F}}_q}(\mathcal{E})$. Then there is a rational function $\dfrac{\phi}{z^{n^\prime}}$ over $\mathbb{\bar{F}}_q$ such that 
\begin{equation}
\sum_{i=1}^k[P_i]-k[\mathcal{O}]=\text{div}\left(\dfrac{\phi}{z^{n^\prime}}\right)
\end{equation}
Bezout's theorem justifies that $\text{deg}(\phi)=n^\prime$, since $\phi$ is zero on $P_1,P_2,\ldots,P_k$. We now claim, there is $\psi$ over $\mathbb{F}_q$ which is also of degree $n^\prime$ and  passes through $P_1,P_2,\ldots,P_k$. First thing to note is that there is a finite extension of $\mathbb{F}_q$, $\mathbb{F}_{q^N}$(say) in which all the coefficients of $\phi$ lies and $\gcd(q,N)=1$. Let $G$ be the Galois group of $\mathbb{F}_{q^N}$ over $\mathbb{F}_q$ and define
\begin{equation}
\psi=\sum\limits_{\sigma\in\mathcal{G}}\phi^\sigma.
\end{equation}
Clearly $\text{deg}(\psi)=n^\prime$. Note that, since $P_i$ for $i=1,2,\ldots,k$ is in $\mathbb{F}_q$ is invariant under $\sigma$. Furthermore, $\sigma$ being a field automorphism, $P_i$ is a zero of $\phi^\sigma$ for all $\sigma \in G$. This proves that $P_i$ are zeros of $\psi$ and then Bezout's theorem shows that these are the all possible zeros of $\psi$ on $\mathcal{E}$. The only thing left to show is that $\psi$ is over $\mathbb{F}_q$. To see that, lets write $\phi=\sum_{i+j+k=n^\prime}a_{ijk}x^iy^jz^k$. Then $\psi=\sum_{i+j+k=n^\prime}\sum_{\sigma\in G}a_{ijk}^\sigma x^iy^jz^k$. However, it is well known that $\sum_{\sigma\in G}a^\sigma\in\mathbb{F}_q$ for all $a\in\mathbb{F}_{q^N}$.

Conversely, if we are given a curve $\mathcal{C}$ of degree $n^\prime$ that passes through $P_1,P_2,\ldots,P_k$. Then consider the rational function $\mathcal{C}/z^{n^\prime}$. Then this function has zeros on $P_i$, $i=1,2,\ldots,k$ and poles of order $k$ at $\mathcal{O}$. The above isomorphism says $\sum_{i=1}^kP_i=\mathcal{O}$. 
\end{proof}

\subsection{How to use the above theorem in our algorithm}
We choose $k$ such that $k=3n^\prime$ for some positive integer $n^\prime$. Then we choose random points $P_1,P_2,\ldots,P_s$ and $Q_1,Q_2,\ldots,Q_t$ such that $s+t=k$ from $\mathcal{E}$ and check if there is a homogeneous curve of degree $n^\prime$ that passes through these points. Where $P_i=n_iP$ and $Q_j=-n_j^\prime Q$ for some integers $n_i$ and $n_j^\prime$. If there is a curve, the discrete logarithm problem is solved. Otherwise repeat the process by choosing a new set of points $P_1,P_2,\ldots,P_s$ and $Q_1,Q_2,\ldots,Q_t$. To choose these points $P_i$ and $Q_j$, we choose a random point $n_i,n_j^\prime$ and compute $n_iP$ and $-n_j^\prime Q$. We would choose $n_i$ and $n_j^\prime$ to be distinct from the ones chosen before. This gives rise to distinct points $P_i$ and $Q_j$ on $\mathcal{E}$. 

The only question remains, how do we say if there is a homogeneous curve of degree $n^\prime$ passing through these selected points? One can answer this question using linear algebra.

Let $C=\sum_{i+j+k=n^\prime}a_{ijk}x^iy^jz^k$ be a \emph{complete} homogeneous curve of degree $n^\prime$. We assume that an ordering of $i,j,k$ is fixed throughout this paper and $C$ is presented according to that ordering.
By complete we mean that the curve has all the possible monomials of degree $n^\prime$. We need to check if $P_i$, $i=1,2,\ldots,s$ and $Q_j$ for $j=1,2,\ldots,t$ satisfy the curve $C$. Note that, there is no need to compute the values of $a_{ijk}$, just mere existence will solve the discrete logarithm problem.

Let $P$ be a point on $\mathcal{E}$. We denote by $\overline{P}$ the value of $C$ when the values of $x,y,z$ in $P$ is substituted in $C$. In other words, $\overline{P}$ is a linear combination of $a_{ijk}$ with the fixed ordering. Similarly for $Q$s. We now form a matrix $\mathcal{M}$ where the rows of $\mathcal{M}$ are $\overline{P_i}$ for $i=1,2,\ldots,s$ and $\overline{Q_j}$ for $j=1,2,\ldots,t$. If this matrix has a non-zero left-kernel, we have solved the discrete logarithm problem. By \emph{left-kernel} we mean the kernel of $\mathcal{M}^\text{T}$, the transpose of $\mathcal{M}$.

\subsection{Why look at the left-kernel instead of the kernel} In this paper, we will use the left-kernel more often than the (right)kernel of $\mathcal{M}$. We denote the left-kernel by $\mathcal{K}$ and kernel by $\mathcal{K}^\prime$. We first prove the following theorem:
\begin{theorem}
The following are equivalent:
\begin{description}
\item[(a)] $\mathcal{K}=0$.
\item[(b)] $\mathcal{K}^\prime$ only contain curves that are a multiple of $\mathcal{E}$.
\end{description}
\end{theorem} 
\begin{proof}
The proof uses a simple counting argument. First recall the well-known fact that the number of monomials of degree $d$ is $d+2\choose 2$. Furthermore, notice two things -- all multiples of $\mathcal{E}$ belongs to $\mathcal{K}^\prime$ and the dimension of that vector-space (multiples of $\mathcal{E}$) is ${n^\prime-1\choose 2} =\dfrac{(n^\prime-2)(n^\prime-1)}{2}$, where $n^\prime$ is as defined earlier.

Now, $\mathcal{M}$ was as defined earlier, has $3n^\prime$ rows and $\dfrac{(n^\prime+1)(n^\prime+2)}{2}$ columns. Then $\mathcal{K}=0$ means that the row-rank of $\mathcal{M}$ is $3n^\prime$. So the dimension of the $\mathcal{K}^\prime$ is 
\begin{equation*}
\dfrac{(n^\prime+1)(n^\prime+2)}{2}-3n^\prime=\dfrac{(n^\prime-2)(n^\prime-1)}{2}.
\end{equation*}
This proves $(a)$ implies $(b)$.

Conversely, if $\mathcal{K}^\prime$ contains all the curves that are a multiple of $\mathcal{E}$ then its dimension is at least $\dfrac{(n^\prime-2)(n^\prime-1)}{2}$, then the rank is $3n^\prime$, making $\mathcal{K}=0$. 
\end{proof}

It is easy to see, while working with the above theorem $\mathcal{M}$ cannot repeat any row. So from now onward we would assume that $\mathcal{M}$ has no repeating rows. For all practical purposes this means that we are working with distinct(unique) partitions.
 
A question that becomes significantly important later is, instead of choosing $k$ points from the elliptic curve what happens if we choose $k+l$ points for some positive integer $l$. The answer to the question lies in the following theorem.

\begin{theorem}
  If $l \geq 1$, the dimension of the left kernel of $\mathcal{M}$ is $l$.
\end{theorem}
\begin{proof}
  First assume $l \geq 1$. In this case, any non-trivial element of
  $\mathcal{K}^\prime$ will define a curve which passes through more than $3n^\prime$ point of the elliptic curve. Since the elliptic curve is irreducible, it must be a
  component of the curve. Thus the equation defining the curve must be
  divisible by the equation defining the elliptic curve. Thus, the dimension
  of $\mathcal{K}^\prime$ is the dimension of all degree $n^\prime$ homogeneous polynomials which
  are divisible by the elliptic curve. This is the same is the dimension of
  all degree $n^\prime - 3$ homogeneous polynomials. Thus, we get
  \begin{equation*}
	\text{dim}(\mathcal{K}^\prime) = \frac{(n^\prime-2)(n^\prime-1)}{2}.
  \end{equation*}
  On the other hand, by rank-nullity theorem, it follows:
  \begin{eqnarray*}
	&\text{dim}(\mathcal{K}^\prime) + \text{dim}(\text{image}(\mathcal{M})) = \frac{(n^\prime-2)(n^\prime-1)}{2} \\  
	&\text{dim}(\mathcal{K}) + \text{dim}(\text{image}({\mathcal{M}}^\text{T})) = 3n^\prime + l.
  \end{eqnarray*}
  Thus, since row rank and the column rank of a matrix are equal,
  \begin{equation*}
	\text{dim}(\mathcal{K}) = 3n^\prime +l-\frac{(n^\prime-2)(n^\prime-1)}{2} + \text{dim}(\mathcal{K}^\prime) = l.
  \end{equation*}
\end{proof}
\begin{corollary}
Assume that $\mathcal{M}$ has $3n^\prime+l$ rows, computed from the same number of points of the elliptic curve $\mathcal{E}$. If there is a curve $\mathcal{C}$ intersecting $\mathcal{E}$ non-trivially in $3n^\prime$ points among $3n^\prime+l$ points, then there is  a vector $v$ in $\mathcal{K}$ with at least $l$ zeros. Conversely, if there is a vector $v$ in $\mathcal{K}$ with at least $l$ zeros, then there is a curve $\mathcal{C}$ passing through those $3n^\prime$ points that correspond to the non-zero entries of $v$ in $\mathcal{M}$. 
\end{corollary}
\begin{proof}
Assume that there is a non-trivial curve $\mathcal{C}$ intersecting $\mathcal{E}$ in $3n^\prime$ points. Then construct the matrix $\mathcal{M}^\prime$ whose rows are the points of intersection. Then from the earlier theorem we see that $\mathcal{K}$ for this matrix $\mathcal{M}^\prime$ is non-zero. In all the vectors of $\mathcal{K}$ if we put zeros in the place where where we deleted rows then those are element of the left kernel of $\mathcal{M}$. It is clear that these vectors will have at least $l$ zeros.

Conversely, if there is a vector with at least $l$ zeros in $\mathcal{K}$, then by deleting $l$ zeros from the vector and corresponding rows from $\mathcal{M}$ we have the required result from the theorem above. 
\end{proof}
\subsection{Veronese embedding and our algorithm}
There is an alternate way of looking at our algorithm through Veronese embedding~\cite[Page 21: Example 2.4 ]{harris}. We present that in this section.

We know that the sum of $3n^\prime$ points $P_1, P_2, \dotsc, P_{3n^\prime}$ on an
  elliptic curve $\ec$, embedded in $\prpl$, is zero if and only if there
  exists a curve $C$ of degree $n^\prime$ in $\prpl$ such that the algebraic-geometric intersection $C \cap \ec$ is the set $\setl{P_1, \dotsc, P_{3n^\prime}}$, counting
  multiplicity. Given a collection of points $\mathcal{P} = \setl{P_1, P_2,
	\dotsc, P_{3n^\prime+l}}$ on the elliptic curve, we need to find some subset that
	has sum zero. To find this subset, we try to find a curve of degree $n^\prime$ which passes
	through $3n^\prime$ points of $\mathcal{P}$. This can be thought of in the following way
	in terms of the Veronese embedding.

  Recall that the Veronese embedding $\vembd : \prpl \to \prsp{D}$ where
  $D = \frac{(n^\prime + 1)(n^\prime + 2)}{2}$, is given by
$\vembd (x_0 : x_1 : x_2) = (z_1 : z_2 : \dotsb : z_D)$, where $z_i = x_1^{a^i_1} x_2^{a^i_2} x_3^{a^i_3}$ for some bijection
  \begin{align*}
	\Phi: \set{k \in \Z}{1 \leq k \leq D} &\rightarrow 
	\set{(n_1, n_2, n_3) \in \N^3}{n_1 + n_2 + n_3 = n^\prime}. \\
	i &\mapsto (a^i_1, a^i_2, a^i_3)
  \end{align*}

  We claim that a curve passes through $3n^\prime$ points $\left\{P_{m_i},\;1 \leq i \leq 3n^\prime\right\}$
  if and only if $\vembd(P_{m_i})$ lie in a hyperplane $H$ of $\prsp{D}$.
  First, suppose that the curve of degree $n^\prime$, given by the equation $\sum_{i, j, k: i + j + k = n^\prime}^{} c_{ijk} x_1^i x_2^j x_3^k = 0$.
  Consider, the hyperplane $H$ given by the equation
  \begin{equation*}
	H(z_1, \dotsc, z_D) = \sum_{i=1}^{D} c_{\Phi(i)} z^D.
  \end{equation*}
  It is clear that $\vembd(P_{m_i}) \in H$. On the other hand, if
  $\vembd(P_{m_i}) \in H = \sum_{i = 1}^D h_i z_i$, they lie on the
  curve $\sum_{i = 1}^D h_i x_1^{a_1^i} x_2^{a_2^i} x_3^{a_3^i}$ where
  $(a_1^i, a_2^i, a_3^i) = \Phi(i)$ as above.

  To put it in an algebraic-geometric context, let $v$ be the composition
  \begin{equation*}
	\xymatrix{
	  \ec\ \ar@{>->}[r] \ar@/_2pc/[rr]_{v} & \prpl \ar[r]^{\vembd} &
	  \prsp{D}.
	}
  \end{equation*}
  The intersection of $\ec$ with a curve of degree $n^\prime$ corresponds to the
  zeroes of a section of a degree $n^\prime$ line bundle. Any such line bundle is
  the pull-back of a degree $1$ line bundle on $\prsp{D}$ via the Veronese
  map $\vembd$. The $H$, as defined above, defines the degree $1$ divisor
  corresponding to this line bundle on $\prsp{D}$. Thus, the problem of
  finding which $3n^\prime$ points among a collection of points $\mathcal{P}$ on an
  elliptic curve lie on a degree $n^\prime$ curve reduces to finding which $3n^\prime$
  points in the image $\vembd(\mathcal{P})$ lie on a hyperplane. The latter
  is the linear algebra problem that we are interested in.
\section{The main algorithm -- reducing ECDLP to a linear algebra problem (Problem L)}
The algorithm that we present in this paper has two parts. One reduces it to a problem in linear algebra and the other solves that linear algebra problem which we call Problem L. The first algorithm, Algorithm 1, is Las Vegas in nature with high success probability. Furthermore, the algorithm is polynomial time in both time and space complexity.

\begin{algorithm}[H]
\DontPrintSemicolon
\KwData{Two points $P$ and $Q$, such that $mP=Q$}
\KwResult{$m$}
Select a positive integers, $n^\prime$ and $l=3n^\prime$. Initialize a matrix with $3n^\prime+l$ rows and ${n^\prime+2\choose 2}$ columns. Initialize a vector $\mathcal{I}$ of length $3n^\prime-1$ and another vector $\mathcal{J}$ of length $l+1$. Initialize integers $A,B=0$. \;
\Repeat{$\mathcal{K}$ has a vector $v$ with $l$ zeros (Problem L)}{
\For{$i=1$ to $3n^\prime-1$}{
\Repeat{$r$ is not in $\mathcal{I}$}{
choose a random integer $r$ in the range $[1,p)$\;
\:}
$\mathcal{I}[i]\leftarrow r$\;
compute $rP$\;
compute $\overline{rP}$\;
insert $\overline{rP}$ as the $i\textsuperscript{th}$ row of the matrix $\mathcal{M}$
\; }
\For{$i=1$ to $l+1$}{
\Repeat{$r$ is not in $\mathcal{J}$}{
choose a random integer $r$ in the range $[1,p)$\;
\:}
$\mathcal{J}[i]\leftarrow r$\;
compute $-rQ$\;
compute $\overline{-rQ}$\;
insert $\overline{-rQ}$ as the $(3n^\prime+i-1)\textsuperscript{th}$ row of the matrix $\mathcal{M}$
\; }
compute $\mathcal{K}$ as the left-kernel of $\mathcal{M}$
\;}  
\For{$i=1$ to $3n^\prime-1$}{
\If {$v[i]\neq 0$}{
$A=A+\mathcal{I}[i]$\;}
\;}
\For{$i=3n^\prime$ to $3n^\prime+l$}{
\If {$v[i]\neq 0$}{
$B=B+\mathcal{J}[i-3n^\prime+1]$\;}
\;}
\Return{$A\times B^{-1}\bmod p$}
\caption{Reducing ECDLP to a linear algebra problem (Problem L)}
\end{algorithm}
\subsubsection{Why is this algorithm better than exhaustive search}
In the exhaustive search we would have picked a random set of $3n^\prime$ points and then checked to see if the sum of those points is $Q$. In the above algorithm we are taking a set of $3n^\prime+l$ points and then checking all possible $3n^\prime$ subsets of this set simultaneously. There are ${3n^\prime+l\choose l}$ such subsets. This is one of the main advantage of our algorithm.
\subsubsection{Probability of success of the above algorithm} 
To compute the probability, we need to understand the number of unique partitions of an integer $m$ modulo a prime $p$. For our definition of partition, order of the parts does not matter.
The number of partitions is proved in the following theorem:
\begin{theorem}
Let $k$ be an integer greater than $2$. The number of $k$ unique partitions of $m$ modulo a odd prime $p$ is $\dfrac{(p-1)(p-2)\ldots(p-k+2)(p-k)}{k!}$. 
\end{theorem}
\begin{proof}
The argument is a straight forward counting argument. We think of $k$ parts as $k$ boxes. Then the first box can be filled with $p-1$ choices, second with $p-2$ choices as so on. The last but one, $k-1$ box can be filled with $p-k+1$ choices. When all $k-1$ boxes are filled then there is only one choice for the last box, it is $m$ minus the sum of the other boxes. So it seems that the count is $(p-1)(p-2)\ldots(p-k+1)$ choices.

However there is a problem, the choice in the last box might not be different from the first $k-1$ choices. To remove that possibility we remove a choice from the last but one box. That choice is $m$ minus the sum of the first $k-2$ boxes divided by $2$.

Since order does not matter, we divide by $k!$. 
\end{proof}
Consider the event, $m$ is fixed, we pick $k$ integers less than $p$. What is the probability that those numbers form a partition of $m$. From the above theorem, number of favorable events is 
$\dfrac{(p-1)(p-2)\ldots(p-k+2)(p-k)}{k!}$ and the total number of events is ${p\choose k}$. Since for all practical purposes $k$ is much smaller than $p$, we approximate the probability to be $\frac{1}{p}$.

Now we look at the probability of success of our algorithm. In our algorithm we choose $3n^\prime$ points from $3n^\prime+l$ points. This can be done in ${3n^\prime+l\choose l}$ ways. Then the probability of success of the algorithm is $1-\left(1-\frac{1}{p}\right)^{{3n^\prime+l\choose l}}$.

Let us first look at the $\left(1-\frac{1}{p}\right)^p$. It is well known that $\left(1-\frac{1}{p}\right)^p$ tends to $\frac{1}{e}$ when $p$ tends to infinity. So if we can make ${3n^\prime+l\choose l}$ close to $p$,  then we can claim the asymptotic probability of our algorithm is $1-\frac{1}{e}$ which is greater than $\frac{1}{2}$.

Since we are dealing with matrices, it is probably the best that we try to keep the size of it as small as possible. Note that the binomial coefficient is the biggest when it is of the form ${2n\choose n}$ for some positive integer $n$. Furthermore, from Stirling's approximation it follows that for large enough $n$, ${2n\choose n}\approx \frac{4^n}{\sqrt{\pi n}}$.

So, when we take $3n^\prime=l$ and such that ${3n^\prime+l\choose l}=p$ then $l$ is the solution to the equation $l=O(1)+\log{l}+\log{p}$.

To understand the time complexity of this algorithm (without the linear algebra problem), the major work done is finding the kernel of a matrix. Using Gaussian elimination, there is an algorithm to compute the kernel which is cubic in time complexity. Thus we have proved the following theorem:
\begin{theorem}
When $p$ tends to infinity, the probability of success of the above algorithm is approximately $1-\frac{1}{e}\approx 0.6321$. The size of the matrix required to reach this probability is O$(\log{p})$. This makes our algorithm polynomial in both time and space complexity.
\end{theorem}
\subsection{Few comments}
\subsubsection{Accidentally solving the discrete logarithm problem} It might happen, that while computing $rP$ and $rQ$ in our algorithm, it turns out that for some $r_1$ and $r_2$, $r_1P=r_2Q$. In that case, we have solved the discrete logarithm problem. We should check for such accidents. However, in a real life situation, the possibility of an accident is virtually zero, so we ignored that in our algorithm completely.
\subsubsection{On the number of $P$s and $Q$s in our algorithm} 
The algorithm will take as input $P$ and $Q$ and produce different $P$s and $Q$s and the produce a vector $v$ with $l$ many zeros. If all of these $l$ zeros fall either in the place of $P$s or $Q$s exclusively, then we have not solved the discrete logarithm problem. To avoid this, we have chosen $P$s and $Q$s of roughly same size, with one more $P$ than $Q$. This way the vector $v$ will have atleast one non-zero in the place of both $P$ and $Q$.  
\subsubsection{Allowing, detecting and using multiple intersection points in our algorithm}
One obvious idea to make our algorithm slightly faster: allow
multiplicities of intersection between the curve $C$ and the elliptic
curve $\mathcal{E}$. This will increase the computational complexity. Since the
elliptic curve is smooth at the points one is interested in, one
observes that with high probability the multiplicity of intersection
will coincide with the multiplicity of the point in $C$. This reduces to
checking if various partial derivatives are zero. This can easily be
done by introducing extra rows in the matrix $\mathcal{M}$. Then the
algorithm reduces to finding vectors with zeroes in a particular
pattern. This is same as asking for special type of solutions in
Problem L. However, this has to be implemented efficiently as probability of such
an event occurring is around $1/p$ for large primes $p$.
\section{Dealing with the linear algebra problem}
This paper provides an efficient algorithm to reduce the elliptic curve discrete logarithm problem to a problem in linear algebra. We call it the Problem L.

At this stage we draw the attention of the reader to some similarities that emerge between the most powerful attack on the discrete logarithm problem over finite fields, the index-calculus algorithm, and our algorithm. In an index-calculus algorithm, the discrete logarithm problem is reduced to a linear algebra problem. Similar is the case with our algorithm. However, in our case, the linear algebra problem is of a different genre and not much is known about this problem. In this paper, we have not been able to solve the linear algebra problem completely. However, we made some progress and we report on that in this section. 
\begin{problem}
Let $W$ be a $l$-dimensional subspace of a $n$-dimensional vector space $V$. The vectors in the vectors space are presented as linear sum of some fixed basis of $V$. The problem is to determine, if $W$ contains a vector with $l$ zeros. If there is one such vector, find that vector.
\end{problem}
This problem is connected with the earlier algorithm in a very straightforward way. We need to determine if the left-kernel of the matrix $\mathcal{M}$ contains a vector with $l$ zeros and that is where Problem L must be solved efficiently for the overall algorithm to run efficiently. As is customary, we would assume that the kernel $\mathcal{K}$ is presented as a matrix of size $l\times (3n^\prime+l)$, where each row is an element of the basis of $\mathcal{K}$.

A algorithm that we developed, uses Gaussian elimination algorithm multiple times to solve Problem L. In particular we use the row operations from the Gaussian elimination algorithm. Abusing our notations slightly, we denote  the basis matrix of $\mathcal{K}$ by $\mathcal{K}$ as well. Now we can think of $\mathcal{K}$ to be made up of two blocks of $l\times l$ matrix. Our idea is to do Gaussian elimination to reduce each of these blocks to a diagonal matrix one after the other. The reason that we do that is, when the first block has been reduced to diagonal, every row of the matrix has at least $l-1$ zeros. So we are looking for another zero in some row. The row reduction that produced the diagonal matrix in the first block might also have produced that extra zero and we are done. However, if this is not the case, we go on to diagonalize the second block and check for that extra zero like we did for the first block.

\begin{algorithm}[H]
\DontPrintSemicolon
\KwData{The basis matrix $\mathcal{K}$}
\KwResult{Determine if Problem L is solved. If yes, output the vector that solves Problem L}
\For{i=1 to 2}{
row reduce block $i$ to a lower triangular block\;
check all rows of the new matrix to check if any one has $l$ zeros\;
\If{there is a row with $l$ zeros}{
STOP and return the row\;}
row reduce the lower-triangular block to a diagonal block\;
check all rows of the new matrix to check if any one has $l$ zeros\;
}
\caption{Multiple Gaussian elimination algorithm}
STOP (Problem L not solved)
\end{algorithm}
\section{Complexity, implementation and conclusion}
\subsection{Complexity} We describe the complexity of the whole algorithm in this section. First note that the whole algorithm is the composition of two algorithms, one is Algorithm 1, which has success probability $0.6$ and the other is the linear algebra problem. It is easy to see from conditional probability that the probability of success of the whole algorithm is the product of the probability of success of Algorithm 1 and Algorithm 2.

Let us now calculate the probability of Algorithm 2 under the condition that Algorithm 1 is successful. In other words, we know that Algorithm 1 has found a $\mathcal{K}$ whose span contains a vector with $l$ zeros. What is the probability that Algorithm 2 will find it?

Notice that Algorithm 2 can only find zero if they are in certain positions and the number of such positions is $l^2$. Total number of ways that there can be $l$ zeros in a vector of size $3n^\prime +l$ is ${3n^\prime+l\choose l}$. In our setting we have already assumed that ${3n^\prime+l\choose l}\approx p$. Then the probability of success of the whole algorithm is
\[0.6\times\dfrac{\left(\log{p}\right)^2}{p}. \]
Which is a significant improvement over exhaustive search!

One thing to notice, the probability of success is $1-\left(1-\frac{1}{p}\right)^{{3n^\prime+l\choose l}}$ and in the probability estimate we have ${3n^\prime+l\choose l}$ in the denominator. Furthermore, one observes that in this paper we have taken ${3n^\prime+l\choose l}$ to approximately equal the prime $p$. One can now question our choice and argue, if we took ${3n^\prime+l\choose l}$ to be much smaller than $p$, we might get a better algorithm. Alas, this is not the case, $1-\left(1-\frac{1}{p}\right)^{p^\frac{1}{n}}$ tends to $0$ as $p$ tends to infinity for $n\geq 2$. 
\subsection{Implementation}
We have implemented the algorithm in sage~\cite{sagemath}. Since the complexity of the algorithm is only little better than exhaustive search there is no point in providing details of implementation. However, we would like to mention that the algorithm works flawlessly with elliptic curves on fields of all characteristics.
\subsection{Conclusion}
We conclude this paper by saying that we have found a new genre of attack against the elliptic curve discrete logarithm problem. This attack has some similarities with the well-known index-calculus algorithm. In an index-calculus algorithm, the discrete logarithm problem is reduced to a problem in linear algebra and then the linear algebra problem is solved. However, the similarities are only skin deep as our linear algebra problem in completely new.
\renewcommand{\baselinestretch}{1}
\bibliography{paper}
\Address
\end{document}